\documentclass[preprint,review,12pt]{elsarticle} % use this to submit to CMA

\usepackage[utf8]{inputenc}

\usepackage[margin=1in]{geometry}
\usepackage{graphicx}
\usepackage{amssymb,amsmath}
\usepackage{epstopdf}
\usepackage{url}
\usepackage{amsthm}
\usepackage[normalem]{ulem}
\usepackage{dsfont}
\usepackage{amssymb}
\usepackage{graphics}
\usepackage{tikz}
\usetikzlibrary{arrows.meta}
\usepackage{array, multirow, boldline}
\usepackage{multirow}
\usepackage{xcolor}
\usepackage[hidelinks]{hyperref}
\usepackage{color,soul}
\usepackage{mathtools}
\usepackage{tablefootnote}
\usepackage{lineno}

\usetikzlibrary{shapes.geometric,arrows,shapes,snakes,automata,backgrounds,petri}
\tikzstyle{startstop} = [rectangle, rounded corners, minimum width=1cm, minimum height=1cm,text centered, draw=black]
\tikzstyle{io} = [rectangle, rounded corners, minimum width=1cm, minimum height=1cm,text centered, draw=black]
\tikzstyle{process} = [rectangle, minimum width=3cm, minimum height=1cm, text centered, draw=black]
\tikzstyle{arrow} = [thick,->,>=stealth,black]
\tikzstyle{darrow} = [dashed,->,>=stealth, draw=blue!1000]
\tikzstyle{garrow} = [dashed,->,>=stealth, draw=violet!100]
\tikzstyle{barrow} = [dashed,->,>=stealth]
\tikzstyle{init} = [pin edge={thick, ->, black}]
\tikzstyle{int} = [pin edge={to-,thick, ->, black}]
\usepackage{wrapfig,lipsum,booktabs}

\newtheorem{theorem}{Theorem}

\newcommand{\beq}{\begin{equation}}
\newcommand{\eeq}{\end{equation}}

\newcommand{\benum}{\begin{enumerate}}
\newcommand{\eenum}{\end{enumerate}}
\newcommand{\bitem}{\begin{itemize}}
\newcommand{\eitem}{\end{itemize}}

\newcommand{\ds}{\displaystyle}

\newcommand{\barg}{\bar{g}}

\renewcommand{\hbar}{\bar{h}}

\definecolor{igreen}{rgb}{0.0, 0.56, 0}

\journal{Computers Mathematics and Applications}

\begin{document}

\begin{frontmatter}

\title{A Mathematical Model of Opinion Dynamics with Application to Vaccine Denial}

\author[1]{Daniel Cicala}
\ead{cicalad1@southernct.edu}

\author[2]{Yi Jiang}
\ead{yjianaa@siue.edu}

\author[3]{Jane HyoJin Lee}
\ead{jlee7@stonehill.edu}

\author[4]{Kristin Kurianski\corref{cor1}}
\ead{kkurianski@fullerton.edu}

\author[5]{Glenn Ledder}
\ead{gledder@unl.edu}

\cortext[cor1]{Corresponding author}

\affiliation[1]{
    organization={Department of Mathematics, Southern Connecticut State University},
    addressline={501 Crescent St.},
    city={New Haven},
    state={CT},
    postcode={06515},
    country={U.S.A.}}

\affiliation[2]{
    organization={Department of Mathematics and Statistics, Southern Illinois University Edwardsville},
    addressline={6 Hairpin Dr.},
    city={Edwardsville},
    state={IL},
    postcode={62026},
    country={U.S.A.}}

\affiliation[3]{
    organization={Department of Mathematics, Stonehill College},
    addressline={320 Washington St.},
    city={Easton},
    state={MA},
    postcode={02357},
    country={U.S.A.}}

\affiliation[4]{
    organization={Department of Mathematics, Cal State University Fullerton},
    addressline={800 N. State College Blvd.},
    city={Fullerton},
    state={CA},
    postcode={92831},
    country={U.S.A.}}

\affiliation[5]{
    organization={Department of Mathematics, University of Nebraska, Lincoln},
    addressline={203 Avery Hall},
    city={Lincoln},
    state={NE},
    postcode={68588},
    country={U.S.A.}}

\date{May 2, 2024}

%%%%%%%%%%%%%%%%%%%%%%%%%%%%%%%%%%%%%%%%%%%%%%%%%%%%%%%

\begin{abstract}
Public health outcomes can be heavily influenced by the landscape of public opinion; hence, it is important to understand how that landscape changes over time.  For one, opinions on public health issues are responsive to official pronouncements, whether from the governmental or professional medical establishments. Additionally, in today's world of high speed communication, opinion can also be highly responsive to the broadcast opinions of ``influencers'' whose large numbers of followers assure them of a broad reach.  To understand the opinion landscape in a general sense, we develop an ordinary differential equation model for opinion change that is based primarily on attraction to the opinions of prominent sources.  The individual opinion change model is then used to develop a Fokker–Planck-type partial differential equation model for the overall opinion landscape.  This model is shown to have a stable equilibrium solution, and the dependence of the equilibrium solution on key model parameters is illustrated with examples based on opinion regarding vaccination.
\end{abstract}


\end{frontmatter}

\section{Introduction}

In recent years, social media has emerged as a growing alternative to traditional institutions for gathering information \cite{bail2018exposure,zhuravskaya2020political}. Individual people are performing roles previously occupied by mainstream news, the medical establishment, and other areas dominated by experts. Occasionally, an individual becomes sufficiently popular as to genuinely compete with a traditional institution; such individuals are called \emph{influencers} \cite{kata2012anti,ki2019mechanism,zimmermann2022influencers}. Information from these influencers plays an outsized role in the formation of individual opinion \cite{diresta2024invisible}.  Another feature of the current landscape is the increasing disagreement between governmental and professional authorities on the extent to which established medical science should be followed on politically-charged issues such as vaccination requirements and vaccine-related research \cite{abcnews2025desantis, cdc2025immunization, cnn2025fundingcuts, pbs2025rfkjr, hhs2025mrna}. This schism correlates with the public's waning confidence in the efficacy and safety of vaccines \cite{annenberg2023vaccine}. To understand how government policy, expert recommendations, and influencer impact shape public opinion, we introduce an opinion dynamics model that incorporates these features. 

To represent the dynamics of opinion change, we propose two models: an individual opinion model and an aggregate opinion model.  Both models take an abstract view of opinion as a scalar value in a continuous range.  The individual opinion model is a deterministic ordinary differential equation that governs the dynamics of an individual's opinion on some issue in response to the various forces that exert a pull on that opinion.  The aggregate opinion model is a partial differential equation that tracks the population density function of opinion in response to opinion diffusion as well as the forces acting on individual opinion.

\subsection{Comparison to other approaches}

The majority of opinion dynamics models are microscopic, in that information is propagated through the interactions of a finite number of agents distributed in an opinion space, with the opinion distribution updated each time a pair of agents meet \cite{banisch2012agent,schweighofer2020agent,shang2014agent}.  A smaller number of models are macroscopic, in the sense that it is the opinion density of a well-mixed population that is considered.  Models can also be classified by the definition of opinion as a discrete variable \cite{galam2002minority,galam1991towards,holley1975ergodic,sznajd2000opinion} or a continuous one \cite{deffuant2000mixing,hegselmann2002opinion,sayama2020enhanced}. The combination of a macroscopic approach and opinion as a continuous variable allows for the use of a differential equation model rather than an individual-based model, thereby providing a deterministic response of aggregate opinion to societal forces.

Opinion change models generally consider one's opinions to be influenced by the opinions of all other individuals in a subset of the larger community, with no one individual's opinion considered more influential than others.
While several studies have incorporated stubborn agents or zealots into models \cite{klamser2017zealotry,verma2014impact}, only a few have incorporated agents with outsized influence \cite{coculescu2024opinion,helfmann2023modelling}. Moreover, we are aware of none in which opinion change is driven primarily by major actors, such as influencers and official mainstream sources.  Our model focuses on major actors, consistent with studies that have identified prominent influencers as a significant driver for opinion formation \cite{diresta2024invisible}.  We go so far as to omit any explicit role of individual interactions among community members.  This may seem limiting, but it makes for a far simpler and parsimonious treatment of opinion dynamics.  There is little practical difference between one ordinary individual affecting another and the latter's having been affected instead by the influencer from whom the former's opinion is derived. 

A common feature of opinion models is \emph{bounded confidence} \cite{hegselmann2002opinion, lorenz2007continuous, zhao2016bounded}. This is a threshold $d$ such that no opinion change occurs if the distance between a pair of agents in opinion space is larger than $d$. By implementing bounded confidence, such models incorporate homophily, which is the tendency for individuals to be attracted primarily to individuals of similar opinion.  Our model assures homophily without bounded confidence by making the strength of influence exponentially small as the opinion difference grows.   

Some investigators use the analogy of interacting agents as interacting particles to build opinion models that draw from the theory of kinetics \cite{albi2016opinion,bertotti2008discrete,perez2018opinion}. In the limit of large populations, a microscopic kinetic-based opinion model becomes a macroscopic opinion model in the form of a Fokker-Planck equation \cite{during2009boltzmann,furioli2017fokker,toscani2006kinetic}. While our model also uses a Fokker-Planck equation, there are key differences in its derivation. The microscopic kinetic model is inspired by a Boltzmann-type equation from particle physics.  Two agents, each with their own location in opinion space, meet and then update their location by moving closer together via a symmetric compromise model with white noise. From this Boltzmann-type equation, by taking appropriate limits, is derived a Fokker-Planck equation, which is solved numerically, using a linearization framework of Egger and Sch{\"o}berl \cite{egger2010hybrid}. This type of model neglects all outside influences on opinion.

In contrast, our model derivation more closely follows the spirit of Sayama (2020) \cite{sayama2020enhanced}. We begin with a deterministic ordinary differential equation to describe the movement of a single agent, not in response to a random interaction with another agent, but in response to established influencers and media. The resulting Fokker-Planck equation arises as the transport equation with advective flux due to individual opinion movement and diffusive flux in the population as a whole. In addition to the numerical results of the kinetics-based models, we give analytical results about well-posedness and stability.  Analytical results about opinion models are scarce, though some do exist \cite{chazelle2017well,gomez2012bounded}. While the results in Sayama's use a Fokker-Planck equation, they do not include the necessary restriction that the total size of a population must be finite \cite{sayama2020enhanced}. In addition to mathematical results, our model is amenable to studies that identify the specific effects of changes in parameter values.

\subsection{Assumptions Used to Build the Model} \label{sec:assumptions}

In practice, a number of influences can lead people to change their opinion, including communication with other people, pronouncements from authority groups, pronouncements from individuals claiming to be experts, and news media reports of current conditions.  To keep the model simple while still including the most important factors, we assume that opinions change in response to four influences:
\benum
\item
One or more groups of experts, such as the American Medical Association and the Centers for Disease Control on medical issues, each with a consensus opinion;
\item
An aggregate of contrarian influencers, with a single consensus opinion;
\item
A general tendency for opinions to moderate over time in the absence of factors that directly influence opinion change; and
\item
Random drift of opinion.
\eenum

In this description, we assume that there is only a small loss of generality in combining contrarian influencers into a single group, as negative voices on controversial issues tend to have opinions derived from other negative voices.  Similarly, internet influencers whose opinions are similar to a group of experts can be considered simply by increasing the influence strength of those experts.  Internet influencers with a neutral opinion do not focus attention on that neutrality. 

While the contrarian influencers are combined into a single group, we do not always do the same for groups of experts.  With regard to medical issues, various groups of experts tended to have similar opinions prior to 2025; since early in 2025, the governmental medical establishment in the United States and some other countries has become increasingly neutral or contrarian toward vaccination and some other public health issues, while the professional medical establishment have retained their prior opinion.  
An example of the trend toward neutrality is the statement by the US Secretary of Health and Human Services that people ``should do their own research on vaccination'' \cite{stolberg2025kennedy}, which is not suggesting that people should do \emph{scientific} research, but merely that they collect opinions from their preferred sources.  
As of this writing, a panel consisting largely of vaccine skeptics has become responsible for making vaccination recommendations in the United States, indicating a trend for the governmental medical establishment to move beyond neutrality to a contrarian view. For example, the Center for Disease Control and Prevention website now states that a link between vaccines and autism cannot be ruled out, after a longstanding stance that there is no such link \cite{cdc2025link}.

The tendency toward moderation is necessary for the mathematical problem of aggregate opinion dynamics to be well-posed, but it is a plausible influence in any case, and we will take its strength to be relatively small compared to the first two groups. 
Random forces are important, as without them, individual opinions would coalesce to a small number of single opinions.  In our model, we include a diffusion term to the model for aggregate opinion dynamics rather than prescribing random forces for individual opinion dynamics.

The first three forces in our list fall into two broad groups.  The tendency to moderation is a \textbf{global} attractor, meaning that its strength increases with opinion distance.  This is reasonable, as the action of a general force on extreme opinions should be stronger than on moderate opinions.  In contrast, the communities of experts and the aggregate of contrarian influencers are \textbf{local} attractors, meaning that they move individuals toward their position with a strength that rapidly dwindles to zero for subjects whose opinion is far removed from that of the influencer.  This is a reasonable assumption even though it is not strictly accurate.  Some people may react negatively to influences that are far removed from their opinion \cite{bail2018exposure}, but these reactions can be assumed to be weak compared to the positive reaction to influences of like opinion.

\section{Models of Individual and Aggregate Opinion\\ Dynamics}

We consider a simple setting in which an individual's opinion on some issue can be quantified as a scalar value $x$ on a continuum, which may be bounded or may be the whole real line.  Our goal is to build a model for the dynamics of individual opinion change and then use it as the foundation for a model of aggregate opinion change.

\subsection{Dynamics of Individual Opinion Change}

Assume that the evolution of the opinion of individual $i$ is governed by a differential equation
\beq
\dot{x_i}=g(x_i,t;p),
\eeq
where $g_x$ is continuous, $g$ is piecewise continuous in time, and $p$ is a vector of parameters needed for the form of $g$.  Individual opinion should not be arbitrarily far from the neutral opinion $x=0$; therefore, we require there to be extreme opinions $\pm X$ such that 
\beq
\label{grestrictions}
g(x,t)<0 \mbox{ for } x>X, \qquad g(x,t)>0 \mbox{ for } x<-X,
\eeq
with nonzero limits as $x \to \pm \infty$.  This guarantees that any initial value of $x$ will evolve to a value in the interval $[-X,X]$.

In our study, we make the additional assumption that the forces that change opinion act independently, allowing us to prescribe the general form
\beq
g(x_i,t)=-\sum_{j=0}^J \, b_j(t) \cdot k_j(x_i-x_j,t) \cdot (x_i-x_j), 
\eeq
where $b_j>0$ represents the influence strength and $k_j$ represents the influence shape.

%and $k \ge 0$.  
%The term $-ax_i$ models the tendency of opinions to moderate over time in the absence of other influences and guarantees (with the assumptions on $k_j$) that \eqref{grestrictions} holds.
Each term in the summation represents the combined effect of a collection of influences having opinion $x_j(t)$.  
We make the following assumptions on the function $k$:

\benum
\item
$k(h,t) \ge 0 \: \forall (h,t)$ so that influencers always attract (see Section \ref{sec:assumptions}). 
\item
$k(h,t)$ is even about $h=0$ so that the influence does not depend on the direction of the pull.
\item
$k'(h,t) \le 0$ for $h>0$, so that the influence relative to $b_j h$ does not increase with distance. 
\eenum

Attractors can be further classified into two types: global and local.  
Global attractors act at any distance: $\ds \lim_{h \to \infty} k(h,t) >0$; while local attractors have strength that vanishes with increasing distance: $\ds \lim_{h \to \infty} k(h,t) =0$.
In the case of local attractors, we can normalize $k$ by requiring
\benum
\setcounter{enumi}{3}
\item
$\int_{-\infty}^\infty |h k(h,t)| \,dh = 1$.  
\eenum
This assures that the coefficient $b_j$ serves as an objective measure of influence strength. Property \eqref{grestrictions} requires that there be at least one global attractor.

\subsection{Dynamics of Aggregate Opinion Change}

Aggregate opinion is determined by an opinion density function $u(x,t)$ such that the probability of a randomly-chosen individual lies in the interval $(x_0,x_1)$ is 
\beq
P \{ x_0<x<x_1 \}=\int_{x_0}^{x_1} u(x,t) \,dx.
\eeq
The evolution of this opinion density function is governed by the transport equation\footnote{See \ref{sec:model-derivation} for a derivation.} 
\beq
\label{transport}
u_t = (Fu)_x, 
\eeq
where
\beq
\label{fluxdef}
Fu \equiv Du_x-gu
\eeq
is the opinion density flux that combines a diffusive component and an advective component, where $g=\dot{x}$.  The idea of the advective component is that individuals of opinion $x$, whose density is $u$, move along the opinion continuum at a rate given by $g$.

The partial differential equation \eqref{transport} has to be augmented by an initial condition and boundary conditions.  In light of property \eqref{grestrictions}, we assume that $u$ is negligible outside of some extreme range $-X<x<X$.\footnote{It will not be 0 because of the small amount of diffusion needed to make the total flux at $\pm X$ be 0.} With the total opinion population fixed at 1, we have
\beq
\label{totalpop}
\int_{-X}^X u(x,t) \,dx=1+o(1),
\eeq
where the last term indicates that the error in ignoring the regions where $|X|>1$ can be made arbitrarily small by increasing $X$.
Differentiating this integral with respect to time yields
\[ 0=\int_{-X}^X u_t \,dx =\int_{-X}^X (Fu)_x \,dx =Fu(X,t)-Fu(-X,t). \] 
Thus, the fluxes at the endpoints must be equal.  Given that the opinion density and its gradient are negligible at $\pm X$, we are justified in interpreting this last result as saying that the fluxes are 0 at the extreme opinions.\footnote{One might have been inclined to use Dirichlet conditions, $u(\pm X,t)=0$; however, this is inconsistent with the requirement that the total opinion population must be conserved \eqref{totalpop}.}  We therefore have boundary conditions
\beq
\label{bcs}
Fu(\pm X, t) = 0.
\eeq

\section{Properties of the Opinion Dynamics Models}

We consider the problem
\beq
\label{prob1a}
u_t = (Fu)_x, \qquad Fu(\pm X,t)=0, \qquad u(x,0)=u_0(x),
\eeq
where
\beq
\label{prob1b}
Fu \equiv Du_x-g(x,t) u, \qquad \int_{-X}^X u_0(x) \,dx = 1, \qquad u_0(x)>0.
\eeq
with $g_x$ and $u_0$ continuous and $g(x,t)$ piecewise continuous in time.  We prove the existence of a unique solution and the stability of its equilibrium distribution.

\subsection{The Equilibrium Distribution}

If there exists a function $\barg$ such that 
\beq
\lim_{t \to \infty} g(x,t)=\barg(x), 
\eeq
then it makes sense to search for an equilibrium solution.
With $u_t=0$, we can immediately integrate the equation $(Fu^*)'=0$ and apply the boundary conditions to obtain a first-order differential equation problem for the equilibrium distribution $u^*(x)$: 
\beq
\label{eqeqn}
D{u^*}'=\barg u^*, \qquad \int_{-X}^X u^*(x) \, dx=1.
\eeq
The ODE in \eqref{eqeqn} represents the physical result that there is no opinion flux \emph{anywhere} when the opinion distribution is at equilibrium. The integral requirement in \eqref{eqeqn} follows from conservation of mass and the integral condition on the initial condition.  Note also that the ODE shows that local extrema of $u^*$ occur precisely at points where $\barg(x)=0$.

The ODE is separable, with solution
\beq
\label{soln1}
u^*(x)=\frac{1}{I_0} \; e^{G(x)},
\eeq
where
\beq
\label{soln2}
G(x) = \frac{1}{D}\int_0^x \barg(y) \,dy, \qquad I_0=\int_{-X}^X e^{G(x)} \,dx.
\eeq
We note for future reference that $u^*>0$ everywhere. 

\subsection{Existence/Uniqueness, and Equilibrium Solution Stability}

\begin{theorem}
The problem \eqref{prob1a}-\eqref{prob1b} has a unique solution.
\end{theorem}

\begin{proof}
The result follows from an application of Theorem 2 in Chapter 5 of Friedman (1964) \cite{friedman2008partial}.
\end{proof}

To prove stability for the autonomous aggregate dynamics problem, we employ the method of separation of variables.

\begin{theorem}
The equilibrium distribution $u^*$ for the aggregate opinion distribution problem \eqref{prob1a}-\eqref{prob1b} with $g(x,t)=\barg(x)$ for $t>t_1$ is globally asymptotically stable.
\end{theorem}

\begin{proof}
Let $u_1(x)=u(x,t_1)$ and $\tau=t-t_1$.  Replacing $t$ with $\tau$, we obtain the autonomous problem with the profile at time $t_1$ serving as the initial condition at time $\tau=0$; hence, we can assume $t_1=0$ without loss of generality. 
Assuming solutions of the form
\beq
u_n(x,t)=\psi_n(t) \phi_n(x)
\eeq
yields the equation
\[ \frac{\psi_n'}{\psi_n}=\frac{[F \phi_n]'}{\phi_n}=-\lambda_n, \]
leading to the separated problems
\beq
\psi_n'=-\lambda_n \psi_n, 
\eeq
and
\beq
\label{phiBCs}
[F \phi_n]'+\lambda_n \phi_n=0, \quad F \phi_n(\pm X)=0.
\eeq

Note that if $\lambda=0$, then the eigenvalue problem is the same as the equilibrium problem; hence, $\lambda=0$ is an eigenvalue with corresponding eigenfunction $u^*$.  More generally, we can rewrite the differential equation for $\phi_n$ as
\beq 
\label{phieqn}
D\phi_n''-\barg \phi_n' - \barg' \phi_n +\lambda_n \phi_n = 0.
\eeq
Multiplying by $1/{u^*}(x)$ and using $D {u^*}'=\barg u^*$ recasts the differential equation as
\beq
\left[ \frac{D}{u^*} \phi_n' \right]'-\frac{\barg'}{u^*} \phi_n +\frac{\lambda_n}{u^*} \phi_n = 0.
%\left[ D{u^*}^{-1} \phi_n' \right]'-{u^*}^{-1} \barg' \phi_n +\lambda_n {u^*}^{-1} \phi_n = 0.
\eeq
The function $u^*$ is continuous and positive on $[-X,X]$ and $\barg'$ is continuous on $[-X,X]$; therefore, the eigenvalue problem is a Sturm-Liouville problem of regular type; hence, Sturm-Liouville theory offers the following results \cite{haberman1983elementary}:
\benum
\item
The eigenvalues are strictly increasing: $\lambda_0 < \lambda_1 < \cdots$.
\item
The set of eigenfunctions is complete.
\item
The eigenfunctions $\phi_n$ have $n$ zeros on the interval $(-X,X)$. %\Yi{Should it be $n-1$?}\Glenn{I start with $\lambda_0$ rather than $\lambda_1$ so that I can say $n$ zeros rather than $n-1$ zeros.}
%\item
%The eigenfunctions are orthogonal, with weight function ${u^*}^{-1}$.
\eenum

Properties 1 and 2 guarantee that the initial condition $u_0$ has a series representation
\[ u_0(x) = \sum_{n=0}^\infty c_n \phi_n(x), \]
which then means that the unique solution of the aggregate dynamics problem can be written as
\beq 
\label{expansion}
u(x,t) = \sum_{n=0}^\infty c_n e^{-\lambda_n t} \phi_n(x).
\eeq

We have already seen that $u^*$ is an eigenvector for the eigenvalue $\lambda=0$.  Because $u^*$ has no zeros, Property 3 shows that it must be the eigenvector for $\lambda_0$.  The requirement of a unit total integral then guarantees $c_0=1$, so the leading term of the expansion is $u^*$, as desired.

One detail remains: We must still show that a distribution that is initially positive remains positive. Suppose, by contradiction, that such an initial distribution results in a point $(\tilde{x},\tilde{t})$ where $u=0$. Then for any time $\hat{t}<\tilde{t}$ there must be a point $\hat{x}$ and a value $0<\epsilon < \min_{-X \leq x \leq X} u_0(x)$ such that $u \left( \hat{x},\hat{t} \right)=\epsilon$ is the local minimum for $u\left( x,\hat{t} \right)$.  Because solutions are smooth in $x$, we must have $u_x\left( \hat{x},\hat{t} \right)=0$ and $u_{xx}\left( \hat{x},\hat{t} \right)>0$.  But then $u_t\left( \hat{x},\hat{t} \right)=Du_{xx}\left( \hat{x},\hat{t} \right)>0$; hence, the local minimum must be increasing at time $\hat{t}$. Thus, $u$ cannot reach 0.
\end{proof}

\subsection{Approach to Equilibrium}
\label{approaches}

While the partial differential equation can be solved numerically, there is some mathematical value to obtaining results for the asymptotic approach to equilibrium.
From \eqref{expansion}, we have
\beq
\label{series}
u = u^* + c_1 e^{-\lambda_1 t} \phi_1(x) + \sum_{n=2}^\infty c_n e^{-\lambda_n t} \phi_n(x).
\eeq
This means that the eigenfunction $\phi_1$ represents the shape of the approach to equilibrium.  For a range of $x$ where the magnitude of $\phi_1$ is relatively large, the approach to equilibrium will be relatively slow. 

To approximate the eigenvalue $\lambda_1$ and its eigenfunction, we define the 1-parameter family of functions 
\beq 
\label{Qeqn}
Q(x,t;r) \equiv e^{rt} (u-u^*).
\eeq
From \eqref{series}, we have the asymptotic approximation
\beq
\label{Qasympt}
Q(x,t;r) \sim c_1 e^{(r-\lambda_1) t} \phi_1(x), \quad t \to \infty;
\eeq
hence, the long-term behavior of $Q$ depends on the relationship between $r$ and $\lambda_1$.  If $r>\lambda_1$, then $Q$ blows up as time increases, while if $r<\lambda_1$, $Q \to 0$ as time increases.  By adjusting the $r$ value and plotting $\|{Q}\|$ vs $t$, we can obtain a modest bracketing interval for $\lambda_1$ and plot $Q/\|{Q}\|$ vs $x$ for some large $t$ to approximate the normalized eigenfunction $\phi_1$. We illustrate an example of this analysis in Section \ref{examples_g}.

\subsubsection{Computation of $\lambda_1$ and $\phi_1$}

The approach to equilibrium can be better approximated by applying numerical methods to the problem of finding eigenvalues and eigenvectors.  To determine $\lambda_1$ and $\phi_1$, we must solve the eigenvalue problem \eqref{phiBCs} with the definition \eqref{fluxdef}.  This problem has infinitely many solution pairs $(\lambda_n,\phi_n)$.

The difficulty in solving the eigenvalue problem is that the coefficient $\lambda$ is to be determined as part of the solution rather than being prescribed.  A version of the shooting method resolves this difficulty \cite{Iserles2008}.  %\footnote{\Glenn{We can reference a book by Iserles (published by Cambridge) for the shooting method.}\Jane{ref. added}} 
We construct a family of initial value problems having the coefficient $\lambda$ as a parameter.  We can use initial conditions at $x=-X$ and run the solution forward or at $x=X$ and run the solution backwards.  Here we describe the procedure for running the solution forwards.

For convenient system notation, we use $R$ for the eigenfunction $\phi$ and $S$ for the flux $F\phi$.  Thus, the initial value problem is 
\beq
\label{shooting}
DR'=gR+S, \quad S'=-\lambda R, \qquad [R,S](-X)=[\pm 1,0],
\eeq
where the initial condition specification $R(-X)=\pm 1$ is chosen without loss of generality, as the final $\phi_1$ will be normalized.  The choice of $+1$ or $-1$ for $R(-X)$ is arbitrary; in practice, it can be chosen to match the approximate eigenfunction plot obtained from using $Q$ to estimate $\lambda_1$ \eqref{Qasympt}.  
With different values of $\lambda$, the initial value problem yields solutions with different values of $S(X)$.  Eigenvalues are then found as the roots of $y(\lambda) \equiv S(X;\lambda)$.  

\section{Example: Opinion Distribution on Vaccination}
\label{examples_g}

In this section, we illustrate the dynamics of our model under two scenarios with somewhat different influencer profiles. For all scenarios, we take $x_0=0$, $b_0=a$. $k_0=1$, thus making the function $g$ into
\beq
g(x_i,t)=-ax_i-\sum_{j=1}^J \, b_j(t) \cdot k_j(x_i-x_j,t) \cdot (x_i-x_j)~.
\eeq
The $-ax$ term is the general attractor needed to meet the requirements for the function $g$.  For numerical simulations, we take $a=0.1$ and $D=0.1$.

In the absence of data on influence strength as a function of opinion distance, we assume a Gaussian form as an empirical model for $k$ for local attractors:
\beq
k(h)=\frac{1}{\sigma^2} e^{-z^2}, \qquad z=\frac{h}{\sigma}.
\eeq

We consider two specific scenarios: 
\benum
\item
As our default scenario, we consider a system with two local attractors: a strong attractor of broad width at $x_j=1$, corresponding to the combined influence of experts, and a moderate attractor of narrow width at $x_j=-1$ corresponding to the combined influence of contrarian public figures and internet personalities: %Kristin: I removed "(pre-2025)" because we site a 2023 study for this case
\beq
\label{examp1}
[x,b,\sigma]_m = [1,1,1], \qquad [x,b,\sigma]_d = [-1,0.7,0.5],
\eeq
where we have used the subscripts $m$ and $d$, representing ``medical establishment'' and ``deniers,'' respectively, rather than numeric subscripts. 
With no real data available, we assume a relative strength of $b_m=1$ for the medical establishment.  Influence is most significant up to one standard deviation from $x_j$, so the assumed standard deviations mean that the influence of the experts is most significant for $0<x<2$ and that of the contrarian's is most significant for $-1.5<x<-0.5$.  The combination of these two attractors exerts only a small influence in the middle range of $-0.5<x<0$.  Outside the range $-1.5<x<2$, the tendency toward moderation will be the dominant influence.  

The relative strength $b_d=0.7$ of the deniers was chosen to match empirical data showing that approximately 18\% of Americans in a 2023 survey believed vaccines to be unsafe \cite{annenberg2023vaccine}.  This is a conservative estimate as there are undoubtedly some people who accept the safety of vaccines but doubt their efficacy.  This is particularly likely for diseases like COVID-19, where vaccinated individuals can still be infected, although their outcomes tend to be much better due to the vaccine.%Kristin: I vote to remove the clause ", as this benefit of vaccination is not well known." because it is a well-known fact among medical professionals
\item
We also consider a ``conflict'' scenario to represent what is already happening in some countries whose governmental medical establishment is diverging from the professional medical establishment \cite{abcnews2025desantis, cdc2025link, cdc2025immunization, pbs2025rfkjr, hhs2025mrna}. In this scenario, the experts (labeled as ``medical'' in scenario 1) partition into two subgroups, a professional medical establishment group (subscript $p$) with the same opinion and standard deviation as the experts in the default scenario but with strength $b_p=0.7$, and a governmental medical establishment group (subscript $g$) of strength $b_g=0.3$ and width $\sigma_g=1.0$ with opinion $x_g=0$, that is,
\beq
\label{examp2}
[x,b,\sigma]_p = [1,0.7,1], \qquad [x,b,\sigma]_g = [0,0.3,1], \qquad [x,b,\sigma]_d = [-1,0.7,0.5].
\eeq
The parameters for the governmental medical establishment were chosen to be conservative estimates to avoid overestimating the likely effects of the changes.  It seems unlikely that the governmental share of the medical community influence is any less than 30\%, and it may well be more if individual physicians follow the governmental recommendations rather than those of their own professional societies or those indicated by scientific reports.  Nor is it unreasonable to make the governmental opinion neutral, in light of funding cuts and recommendation changes that are already occurring \cite{abcnews2025desantis, hhs2025mrna}. 
\eenum

\begin{figure}[ht]
    \centering        
    \includegraphics{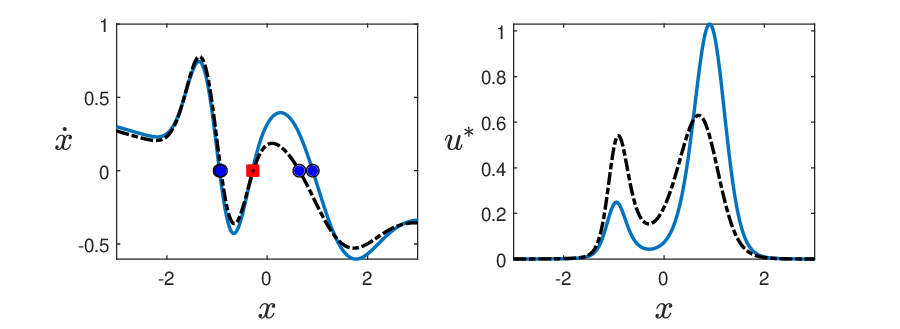}
    \caption{Left: The opinion dynamics function $g(x)=\dot{x}$ for two scenarios, both with moderate local attractors at $x=-1$; stable equilibria are marked with disks, while unstable equilibria are marked with squares.  Right: Equilibrium solutions for the opinion density function. The solid curves are for the scenario described by \eqref{examp1} with a strong local attractor at $x=1$, while the dashed curve represents the scenario described by \eqref{examp2} with moderate local attractors at $x=1$ and $x=0$.}
    \label{ode_fig2}
\end{figure}

The rate of change functions for the two scenarios appear in the left panel of Figure \ref{ode_fig2}, with the first scenario solid and the second scenario dashed.  Both scenarios show two stable equilibria, one positive and one negative, along with an unstable equilibrium that delineates the basins of attraction for the stable equilibria.  The negative stable equilibrium and the unstable equilibrium change only slightly (from $-0.95$ to $-0.92$ and $-0.30$ to $-0.28$); however, the positive stable equilibrium changes significantly from $0.91$ to $0.64$.  In terms of the scenario descriptions, the splitting of the expert opinion does not change the ultimate opinion of those responding to negative influence, nor does it change the direction of movement in terms of the original opinion.  However, it significantly decreases the ultimate opinion of individuals drawn to the more positive attractors to a value representing a compromise between the positive opinion of the professional medical establishment and the neutral opinion of the governmental medical establishment. 

\begin{figure}[ht]
    \centering        
    \includegraphics{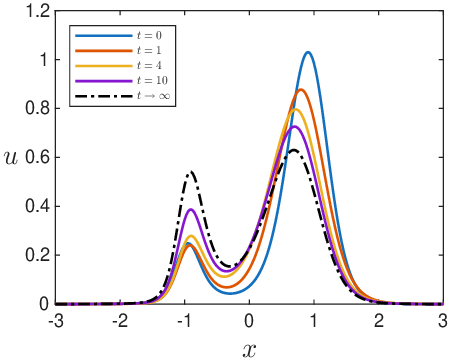}
    \caption{Time series plots for the problem \eqref{prob1a} using the equilibrium solution for the default scenario \eqref{examp1} as the initial condition and the instance of $g$ from the speculative scenario \eqref{examp2} for the partial differential equation and boundary conditions.} 
    \label{pdesim_ex1}
\end{figure}

As an example of the evolution of opinion, we consider a scenario for \eqref{prob1a} in which the equilibrium solution for scenario described by \eqref{examp1} serves as the initial condition and the function $g$ resulting from the parameter values in \eqref{examp2} is used in the differential equation and boundary conditions.  Figure \ref{pdesim_ex1} shows time series plots obtained using MATLAB's pdepe function \cite{matlabcite}.  The shift of the positive peak toward less intensity and lower positive opinion happens relatively quickly compared to the increase in intensity of the negative peak.  The results suggest that the first sign of change will be a decrease of positive opinion among those who are not under negative influence, as seen by a general shift of the right portion of the $u$ vs $x$ curve to the left.  The shift toward a larger number of people with negative opinion is only seen much later.

\begin{figure}[ht]
    \centering            
    \includegraphics{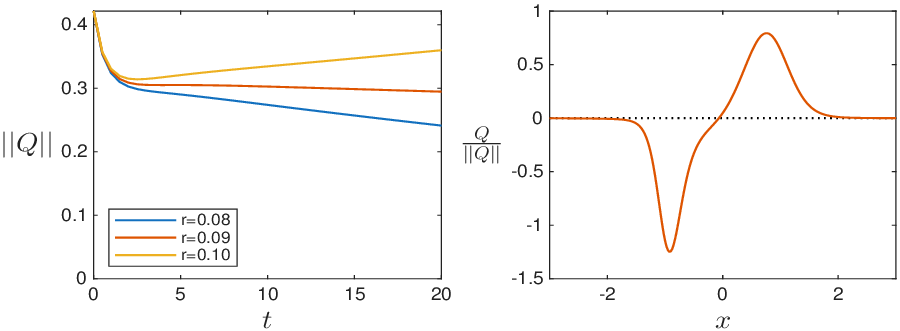}
    \caption{The functions $\|Q\|(t;k)$ and $Q(x,20;0.09)$ from \eqref{Qeqn}; the former shows that $\lambda_1 \approx 0.09$ and the latter shows a visual approximation of $\phi_1$.}
    \label{lambda1_sim}
\end{figure}

To better understand the approach to equilibrium for the example, we employ both methods described in Section \ref{approaches}.  First, we consider the function family $Q(x,t;r)$ defined in \eqref{Qeqn}  The left panel of Figure \ref{lambda1_sim} shows plots of $\|Q\|$ versus time for selected values of $r$.  As noted above, this plot should converge, in theory, to a horizontal asymptote as $t \to \infty$.  In practice, numerical error in the simulation prevents this from happening.  Nevertheless, it is clear from the plot that $\lambda_1$ lies in the interval $(0.08,0.10)$.  A visual approximation of the corresponding eigenfunction $\phi_1$ appears in the right panel of the plot.  This plot is consistent with the observation from the numerical solution of Figure \ref{pdesim_ex1} that convergence to equilibrium is slowest near the distribution peaks. 

\begin{figure}[ht]
    \centering            
    \includegraphics{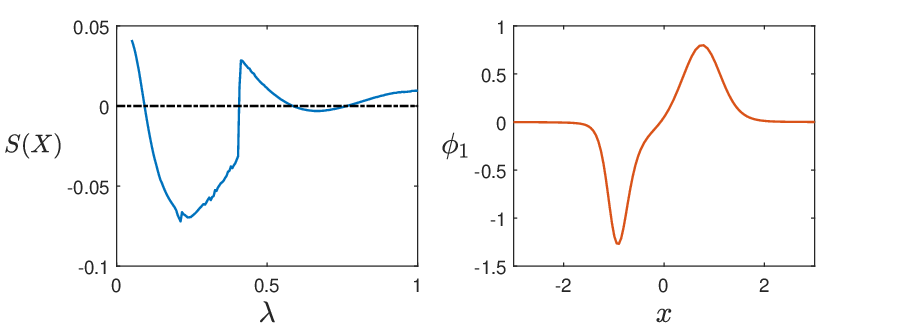}
    \caption{(left) Plot of $S(X)$ as a function of a parameter $\lambda$, where eigenvalues are the roots of this function. (right) The normalized eigenfunction $\phi_1$.}
    \label{lambda1_ode}
\end{figure}

For a better estimate of the values of $\lambda_n$ and the function $\phi_1$, we plot the value achieved by simulation for $S(X)$ in \eqref{shooting} as a function of $\lambda$.  This appears in the left panel of Figure \ref{lambda1_ode}, which shows that there are four eigenvalues less than 1, with $\lambda_1=0.0927$ found numerically as the first root of $S(X)$.  The corresponding eigenfunction appears in the right panel of the figure.

\subsection{Effect of Parameter Values}

\begin{figure}[ht]
    \centering            
\includegraphics[width=.475\textwidth]{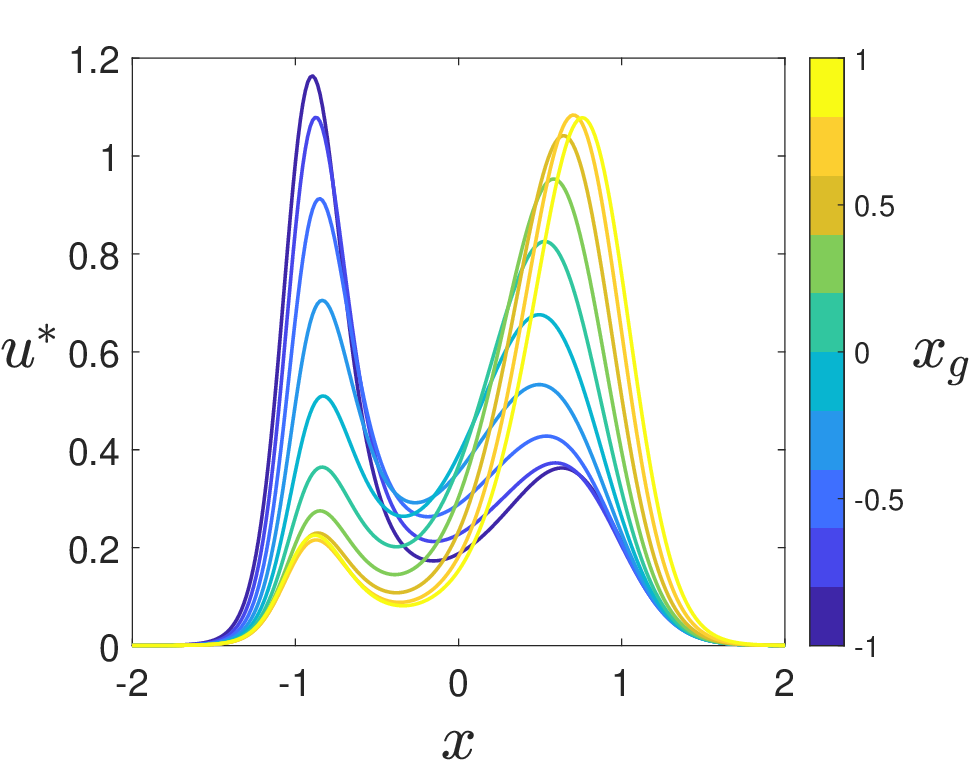}\quad
 \includegraphics[width=.475\textwidth]{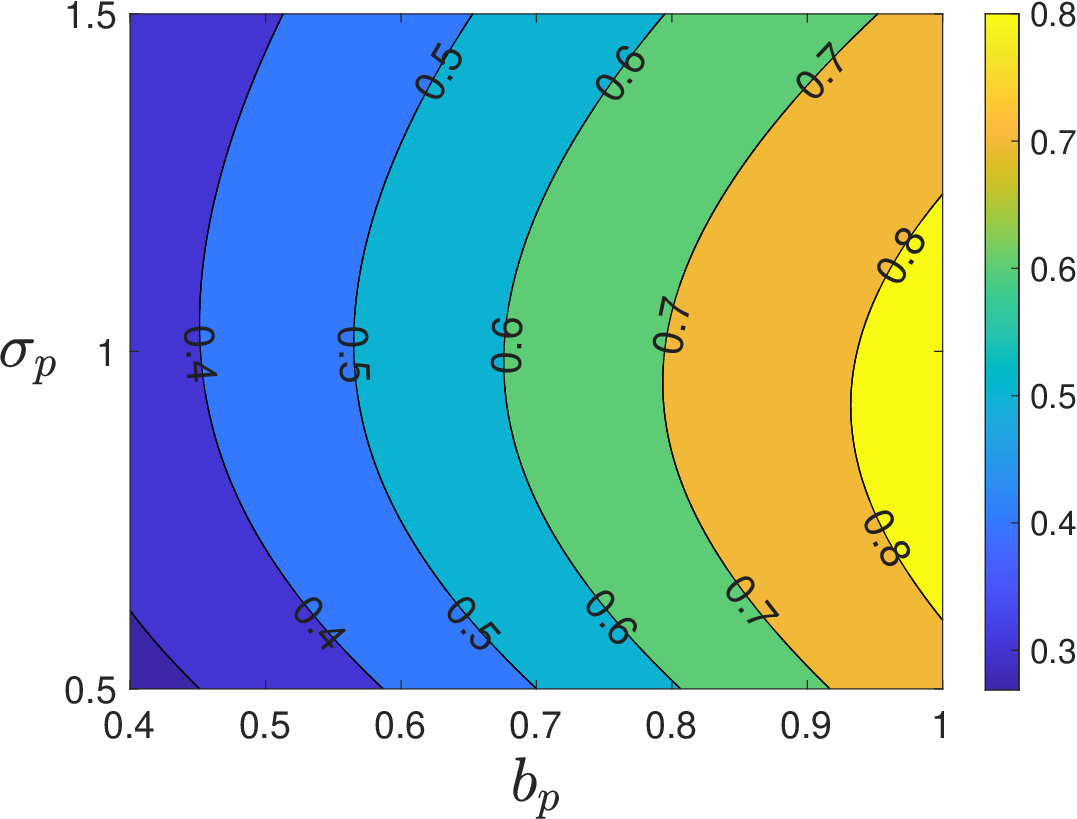}
    \caption{(left) The distribution of opinions at equilibrium with varying $x_g$; (right) level curves of individuals with opinion at least 0 at equilibrium with varying $b_p$ and $\sigma_p$; in both cases the other parameters were set at the 2025 default levels.}
    \label{paramstudies} % From Kristin: Pstarvsbpsigp_v3.eps is an updated version of Pstarvsbpsigp.eps that has larger font on the axes and contour labels
\end{figure}

To assess the effect of the parameter values chosen for the \eqref{examp2} scenario, we consider some studies in which we vary one or more parameters from among those that are likely to be influenced by public health outreach or by current epidemiological conditions, such as the rate at which new infections are being reported.  These are shown in Figure \ref{paramstudies}. 
Rather than examine all the scenario parameters, we consider only $x_g$, $b_p$, and $\sigma_p$.  The first of these is clearly very sensitive to public health conditions, as shown by the large range of recommendations for the measles vaccine offered by state governments as the number of new measles cases waxes and wanes \cite{KDHEmeaslevaccine}. The parameters $b_g$ and $\sigma_g$ are also likely variable in practice, but these should make less of a difference than changes in $x_g$.  The parameter $x_p=1$ is based on medical science and therefore very stable, while the opinions and reach of deniers is far less likely to respond to changes in current conditions than the parameters for the medical establishments.  We do consider variation in $b_p$ and $\sigma_p$, which could change based on outreach from the professional medical establishment.

The left panel of Figure \ref{paramstudies} shows the variation of the equilibrium solution $u^*(x)$ as the opinion value $x_g$ of the governmental medical establishment varies over the full range $[-1,1]$.  Surprisingly, the $x$-values of the peaks are not very sensitive to the opinion of the governmental establishment, but the heights of the peaks are quite sensitive.  As the governmental establishment gradually changes its opinion from positive to negative, more people whose opinions are flexible change from positive to negative.  Opinion is highly polarized for any value of $x_g$, with low population densities for neutral opinion levels.

The right panel shows a contour plot of the function
\beq
P^*_+ = \int_0^X u^*(x) \,dx,
\eeq
which represents the fraction of individuals whose opinion is nonnegative at equilibrium and serves as a simple measure of the overall opinion landscape, as the strength $b_p$ and standard deviation $\sigma_p$ of the professional medical establishment change over a range of values on either side of the default values.  These parameters are likely to be somewhat dependent on public awareness campaigns, with $b_p$ increasing because of increased messaging and $\sigma_p$ increasing because of a wider range of outlets for messaging.  We see that significant changes in the strength parameter $b_p$ can have a large impact on the model outcome.  The impact of the standard deviation parameter $\sigma_p$ is smaller and non-monotone, with the largest population with opinion $x>0$ occurring for a moderate value of $\sigma_p$.  The smaller response suggests that efforts to improve outcomes through greater investment of resources by professional societies should focus on the total messaging with much less concern about the population segments targeted by the messaging.

\section{Discussion}

In this work, we have developed and analyzed a novel macroscopic model for the evolution of opinion driven largely by ``influencers,'' consisting of organizations and collections of individuals who have an outsize influence on others, either because of an official role or a significant number of followers. These influencers attract individuals to their opinions.  This alters the opinion landscape as cohorts of individuals carry their population size with them as they move to new opinions, with the evolution of the opinion distribution governed by the partial differential equation of advective and diffusive transport.  The resulting model has desirable mathematical properties; in particular, there is always a unique solution, and the opinion distribution converges to an equilibrium distribution for scenarios in which all influencer parameters approach steady values.  

From a modeling perspective, the primary interest in the model is to understand how the opinion landscape changes as key influencer properties change.  Our basic vaccination scenario includes three groups of influencers: a professional medical establishment, whose opinion is positive and unchanging but whose influence level and reach might vary; a governmental medical establishment, whose opinion could vary along the spectrum from very positive to very negative; and a collective of deniers, whose opinion is negative and unchanging and whose influence level and reach are assumed to be roughly unchanging.  The principal effect of a change in governmental opinion from positive to negative is to shift pockets of population from opinions similar to the professional medical establishment to opinions similar to the deniers.  This effect is large even though we set the parameter $b_g$ for the strength of the governmental influence to a conservative value. The strength of the professional establishment influence can have a significant effect on the opinion distribution, as seen by the variation in the total population fraction having a nonnegative opinion.  The standard deviation of the professional establishment influence has only a small effect.

One difficulty in using a model such as ours is finding appropriate parameter values.  With only minimal data, one must guess at the values.  Consequently, the results should not be taken as quantitatively meaningful, but the qualitative features of the results can be expected to hold with changes in the base parameters.  A related difficulty is the lack of a measurable time scale.  The time variable $t$ in our model does not have specific units.  This would also require data that is unavailable.  Without such data, the model can predict trends, but not the amount of time required for trends to be seen.

The opinion dynamics model presented here could potentially be coupled with an epidemiology model such as the one in Jiang et al. \cite{jiang2022attitude}, which makes empirical assumptions about how willingness to obtain vaccination changes as the disease incidence changes.  In principle, an opinion dynamics model should be used to model changes in willingness to obtain vaccination.  This could be done by using information from the epidemiology model to change parameter values in the opinion dynamics model and then using some output from the opinion dynamics model to change parameters related to vaccination willingness in the epidemiology model.

\section{Acknowledgments}

The authors would like to thank the Research Experience for Undergraduate Faculty (REUF) program for supporting this work. REUF is a program of the American Institute of Mathematics (AIM) and the Institute for Computational and Experimental Mathematics (ICERM), made possible by the support from the National Science Foundation (NSF) through DMS 1239280. Author Kurianski would also like to thank California State University, Fullerton for their support through the Junior/Senior Faculty Research Grant.

\appendix

\section{Model Derivation} \label{sec:model-derivation}

The forces that change individual opinions also act to change the population distribution of those opinions.  This is described by the advective transport equation, which we derive here, with diffusion to be added later.
The dependent variable of the model is $u(x,t)$, the population density function for a distribution of opinions $x$ on a continuous range of opinion values $S$.  We derive the advective transport equation through a careful accounting of changes in opinion $x$, opinion density $u(x,t)$, and population $u(x,t) dx$ on the infinitesimal time interval $[t,t+dt]$.  

Opinion $x$ changes according to the equation $\dot{x}=g(x,t)$; hence the opinion $x$ at time $t$ changes to the opinion
\beq
\label{newx}
x(t+dt) = x(t)+g(x,t) \,dt + O(dt^2)
\eeq
at time $t+dt$.
Using \eqref{newx}, the opinion density changes to
\beq
\label{newu}
u(x(t+dt),t+dt)=u(x,t)+[u_t(x,t)+u_x(x,t)g(x,t)]dt+ O(dt^2)
\eeq
at time $t+dt$.

To compute the population change for the cohort with initial opinion $[x,x+dx]$, we must first account for the change in the width $dx$ of the opinion interval owing to the different rates of change for the opinions $x$ and $x+dx$.  From \eqref{newx}, the opinion $x+dx$ changes to
\begin{align}
\label{newxplusdx}
\begin{split}
[x+dx](t+dt) &= [x+dx]+g(x+dx,t) \,dt + O(dt^2) \\ 
             &=[x+dx]+[g(x,t)+g_x(x,t) \,dx]\,dt+O(dx^2\,dt+dt^2).    
\end{split}
\end{align}
Subtracting \eqref{newx} yields a new interval width of
\beq
\label{newh}
[1+g_x(x,t)\,dt+O(dx\,dt)] dx.
\eeq
At time $t+dt$, the population that was originally $u(x,t)dx$ is now the product of the new population density \eqref{newu} and the new interval width \eqref{newh}, or (with all quantities evaluated at the original coordinates $(x,t)$)
\beq
\label{newp}
[u+(u_t+u_xg+ug_x) \,dt + O(dx \, dt+dt^2)] \,dx.
\eeq
The total population in the initial cohort is conserved from time $t$ to time $t+dt$; hence, the new population \eqref{newp} has to be the same as the original population $u(x,t) \,dx$.  Subtracting this quantity from the new population, removing the common factors of $dx$ and $dt$, and setting the result equal to 0 yields
\[
(u_t+u_xg+ug_x) + O(dx +dt) =0.
\]
Now taking limits as $dx$ and $dt$ go to 0, we obtain the advective transport equation
\beq
u_t +u_xg+ug_x=0,
\eeq
which we can write as
\beq
u_t=(Fu)_x, \qquad Fu(x,t)=-gu.
\eeq
The quantity $-gu$ is the advective flux of $u$.  The transport equation of opinion dynamics results from adding the diffusive flux $Du_x$ to $Fu$.

\bibliographystyle{elsarticle-harv}
\bibliography{influ-biblio}

\end{document}